\documentclass[11pt, english]{article}
\usepackage{amsmath,amssymb,amsfonts}
\usepackage{paralist}
\usepackage{amsthm,thmtools}
\usepackage{comment}
\usepackage{mathtools}
\usepackage{color}
\usepackage{bbm}
\usepackage{fullpage}

\newtheorem{theorem}{Theorem}[section]
\newtheorem{lemma}[theorem]{Lemma}
\newtheorem{fact}[theorem]{Fact}
\newtheorem{corollary}[theorem]{Corollary}

\newtheorem{prop}[theorem]{Proposition}

\theoremstyle{remark}

\newcommand{\E}{\mathbf{E}}
\newcommand{\R}{\ensuremath{\mathbb{R}}}
\newcommand{\Z}{\ensuremath{\mathbb{Z}}}
\newcommand{\eps}{\epsilon}

\newcommand{\set}[1]{\{{#1}\}}

\title{Lattice-based Locality Sensitive Hashing is Optimal}

\author{
Karthekeyan Chandrasekaran\thanks{University of Illinois, Urbana-Champaign, USA, Email: {\tt karthe@illinois.edu}.}
\and
Daniel Dadush\thanks{Centrum Wiskunde \& Informatica, The Netherlands, Email:
{\tt dndadush@gmail.com}.} 
\and
Venkata Gandikota\thanks{Purdue University, Email: {\tt \{elena-g, vgandiko\}@purdue.edu}.}
\and Elena Grigorescu\footnotemark[1]
} 

\begin{document}
\maketitle
\def\conf{1} 

\thispagestyle{empty}
\begin{abstract}
Locality sensitive hashing (LSH) was introduced by Indyk and Motwani (STOC `98)
to give the first sublinear time algorithm for the $c$-approximate nearest
neighbor (ANN) problem using only polynomial space. At a high level, an LSH family hashes ``nearby'' points to the same bucket and ``far away'' points to different buckets. The quality of measure of an LSH family is its LSH exponent, which helps determine both query time and space usage.
  
In a seminal work, Andoni and Indyk (FOCS `06) constructed an LSH family based on \emph{random ball partitionings} of space that achieves an LSH exponent of $1/c^2$ for the $\ell_2$ norm, which was later shown to be optimal by Motwani, Naor and Panigrahy (SIDMA `07) and O'Donnell, Wu and Zhou (TOCT `14). Although optimal in the LSH exponent, the ball partitioning approach is computationally expensive. So, in the same work, Andoni and Indyk proposed a simpler and more practical hashing scheme based on \emph{Euclidean lattices} and provided computational results using the $24$-dimensional Leech lattice. However, no theoretical analysis of the scheme was given, thus leaving open the question of finding the exponent of lattice based LSH.

In this work, we resolve this question by showing the existence of lattices
achieving the optimal LSH exponent of $1/c^2$ using techniques from the geometry
of numbers. At a more conceptual level, our results show that optimal LSH space
partitions can have \emph{periodic structure}. Understanding the extent to which additional structure can be imposed on these partitions, e.g.~to yield low space and query complexity, remains an important open problem.
\end{abstract}

\section{Introduction}

Nearest neighbor search (NNS) is a fundamental problem in data structure
design. Here, we are given a database $P$ of $n$ points in a metric space $X$, and the goal is
to build a data structure that can quickly return a closest point in the
database to any queried target. In its exact form, the
problem is known to suffer from the curse of dimensionality, where data
structures that beat brute force search (i.e. a linear scan through the data
points) require either space or query time exponential in the dimension of the space
$X$. To circumvent this issue, Indyk and Motwani~\cite{conf/stoc/IM98} studied a
relaxed version of NNS which allowed for both \emph{approximation} and
\emph{randomization}. In $(c,r)$-approximate nearest neighbor search
(ANN), we are given an approximation factor $c \geq 1$ and distance
threshold $r > 0$, where we must guarantee that for a query $q$, if $d_X(q,P)
\leq r$ then the data structure returns $p \in P$ such that $d_X(q,p) \leq cr$.
When we allow randomization, we only require that any fixed query succeeds with
good probability over the randomness used to construct the data structure. 

In order to address ANN, Indyk and Motwani introduced the concept of Locality Sensitive
Hashing (LSH). A locality sensitive hash function maps ``nearby'' points together and ``far away'' points apart. Indyk and Motwani showed that such LSH function families can be used to build data structures with both sublinear query time and subquadratic space for ANN. LSH is now one of the most popular methods for solving ANN and has found many applications in areas such as cryptanalysis~\cite{conf/crypto/Laarhoven15,conf/soda/BDGL16}, information retrieval and machine learning (see~\cite{SDI06} for a survey).  Important metric spaces for LSH include $\set{0,1}^d$ or $\R^d$ under $\ell_1$ or $\ell_2$-norms, and the sphere $S^{d-1}$ under angular distance. In this work, we focus on $\R^d$ under the $\ell_2$-norm. 

Let $\mathcal{H}$ be a family of functions with an associated probability distribution. An LSH family $\mathcal{H}$ is $(c,r,p_1,p_2)$-sensitive for $X$ if a randomly chosen hash function $h$ from $\mathcal{H}$ maps any two points in $X$ at
distance at most $r$ to the same bucket with probability at least $p_1$ and any two points in $X$ 
at distance at least $cr$ to the same bucket with probablity at most $p_2$. The
measure of quality of the LSH family is the so-called LSH exponent $\rho :=
\ln(1/p_1)/\ln(1/p_2)$. If $X = (\R^d,\ell_2)$ and the maximum computational time for evaluating the hash function $h(x)$ at any point $x\in X$ for any element $h\in \mathcal{H}$ is at most $\kappa$, then one can build a randomized $(c,r)$-ANN data structure that answers 
queries in $O((d+\kappa)n^{\rho(c)}\log_{1/p_2}(n))$ time using $O(dn +
n^{1+\rho(c)})$ space~\cite{conf/stoc/IM98,jour/toc/HIM12}. Similar results hold for other
$d$-dimensional metric spaces. Consequently, much research effort has been directed at
constructing LSH families with both low LSH exponent and fast evaluation times. 

For the $\ell_2$-norm, the first results~\cite{conf/stoc/IM98,conf/vldb/GIM99} gave
constructions achieving an exponent $1/c \pm o(1)$ for $X$ being the hypercube
$\set{0,1}^d$, which was later extended to all of $\R^d$
in~\cite{conf/socg/DIIM04}. For the $\ell_2$-norm over $X=\R^d$, Andoni and Indyk~\cite{conf/focs/AndoniIndyk06} gave the first construction of an LSH hash family achieving a limiting
exponent of $1/c^2$, which was later shown to be optimal
in~\cite{jour/siamjdm/MNP07,journ/acmtoct/OWZ14}. We note that optimality here
holds only for ``classical'' LSH, in which the LSH family depends only on the
ambient metric space and not on the database itself, and that these lower bounds
have been recently circumvented using more sophisticated data dependent
approaches~\cite{conf/soda/AINR14,conf/stoc/AR15}, which we discuss later.

While achieving the optimal exponent, the hash functions from Andoni and Indyk's work ~\cite{conf/focs/AndoniIndyk06} are unfortunately quite expensive to evaluate. Their hash function family can be described as
follows: For a design dimension $k$, a function from the family 
%is indexed by a ``random ball covering'' of $\R^k$ corresponding to
corresponds to $k^{O(k)}$ random shifts
$t_1,t_2,\dots$ of the integer lattice $\Z^k$ which satisfy that every point in
$\R^k$ is at $\ell_2$ distance at most $1/4$ from at least one shift. To map the
database and the queries into $\R^k$, the hash function uses a Gaussian
random projection $G$ mapping $\R^d$ to $\R^k$. The hash value on query $q$ then
equals the closest vector to $Gq$ in $\Z^k+t_i$, where $i$ is the first index
such that $Gq$ is at distance at most $1/4$ from some point in $\Z^k+t_i$. For this family
they prove an upper bound on the LSH exponent of $1/c^2 + O(\log k/\sqrt{k})$, which
tends to $1/c^2$ as $k \rightarrow \infty$. Note that storing the description
of this hash function requires $k^{O(k)}$ space and evaluating it requires
iterating over all shifts which takes $k^{O(k)}$ time. This prohibitive space
usage and running time restricted the use of these hash functions to only very
low dimensions in the context of ANN (i.e. $k$ is restricted to be a very
slow growing function of the number of points $n$ in the database),  
yielding a rather slow convergence to the optimal
$1/c^2$ exponent. \\

%\paragraph
\noindent \textbf{Lattice based LSH.} Motivated by the above-mentioned drawbacks,
Andoni and Indyk \cite{conf/focs/AndoniIndyk06} proposed a simpler and more
practical LSH scheme based on \emph{Euclidean Lattices}. A $k$-dimensional
lattice $L \subset \R^k$ given by a collection of basis vectors
$B=(b_1,\dots,b_k)$ is defined to be all integer linear combinations of
$b_1,\ldots, b_k$. The \emph{determinant} of $L$ is defined as $|\det(B)|$,
which we note is invariant to the choice of basis. In lattice based LSH, one
simply replaces the $k^{O(k)}$ shifts of $\Z^k$ by a single random shift $t\in
\R^k$ of a lattice $L$, and the hash value on query $q$ now becomes the closest
vector to $Gq$ in $L+t$. 

We note that the last step of the hashing algorithm
corresponds to solving the \emph{closest vector problem} (CVP) on $L$, i.e. given a
target point $q$ one must compute a closest vector to $x$ in $L$ under the
$\ell_2$ norm.  While this problem is NP-Hard in the worst
case~\cite{jour/mor/Kannan87}, in analogy to coding, one has complete freedom
to \emph{design the lattice}. Thus the main potential benefit of lattice based
LSH is that one may hope to find ``LSH-good'' lattices (i.e., lattices with good LSH exponent) for which CVP can be
solved quickly (at least much faster than enumerating over a ball partition). A
secondary benefit is that the corresponding hash functions require very little
storage compared to the ball partitions, namely just a single shift vector $t$
together with the projection matrix $G$ are sufficient (note that the lattice is
shared across all instantiations of the hash function). To evaluate lattice based 
LSH, Andoni and Indyk ~\cite{conf/focs/AndoniIndyk06} provided experimental results for $L$ being
the 24 dimensional Leech lattice equipped with the decoder
of~\cite{jour/ieeetoc/AB96}. A version of this scheme with the $8$ dimensional
E8 lattice has also been implemented and tested in~\cite{conf/icassp/JASG08},
and a parallelized GPU implementation of the Leech lattice scheme was tested
in~\cite{conf/ieeespdp/ACW13}. 

The following natural question was left open in the work of Andoni and Indyk:
can the space partitions induced by lattices achieve the optimal LSH constant
for the $\ell_2$-norm? Note that for a lattice $L$, the associated space
partition corresponds to a random shift of the tiling of space $\set{y +
\mathcal{V}_L: y \in L}$, where $\mathcal{V}_L$ is the
\emph{Voronoi cell} of the lattice, i.e.~the set of all points closer to the
origin than to any other lattice point.\\

\begin{comment}
As no theoretical analysis of lattice based LSH schemes has yet been given, the
following natural question was left open: can the space partitions induced by
lattices achieve the optimal LSH constant for the $\ell_2$ metric? Note that for
a lattice $L$, the associated space partition corresponds to a random shift of the tiling
of space $\set{y + \mathcal{V}_L: y \in L}$, where $\mathcal{V}_L$ is the
\emph{Voronoi cell} of the lattice, i.e.~the set of all points closer to the
origin than to any other lattice point.
\end{comment}

%\paragraph
\noindent{{\bf Our Contribution.}} As our main result, we resolve this question
in the affirmative. We show that for any fixed approximation factor $c > 1$,
there exists a sequence of lattices $\set{L_{k,c} \subset \R^k: k \geq 1}$,
where $L_{k,c}$ has an associated LSH exponent for $\ell_2$-norm bounded by $1/c^2
+O(1/\sqrt{k})$. We note that this is slightly better than the rate of
convergence to optimality proven by Andoni and Indyk in~\cite{conf/focs/AndoniIndyk06} for the ball partitioning approach. To prove this result, we rely on the probabilistic
method, using a delicate averaging argument over the space of all lattices of
determinant $1$. 

Our result is currently non-constructive, as we lack the appropriate
concentration results for the LSH collision probabilities, though we believe
this should be achievable. A simple and efficient sampling algorithm for the
random lattice distribution that we employ -- known as the Siegel measure over
lattices -- was given by Ajtai~\cite{conf/focs/Ajtai02}, and we expect that a
lattice sampled from this distribution should be ``LSH-good'' (in terms of the
LSH exponent) with high probability. Perhaps a more significant issue is that
for the same dimension $k$, the probabilistic argument may produce different
lattices for different approximation factors. Resolving this issue would require
a much finer understanding of the shape of the collision probability curve
(currently, we can only control the curve at two points), and we leave this as
as an open problem. We note however, that if one allows for sampling a different
random lattice for each hash function instantiation, as opposed to a single
lattice shared by all instantiations, then our methods are indeed constructive.
We find this approach somewhat less appealing however, since in general the cost
of preprocessing a lattice in the context of CVP, say computing a short basis,
the Voronoi cell, etc., is substantial, and hence it is desirable to only have
to perform such preprocessing once. Furthermore, since the end goal is
eventually to find a class of LSH good lattices with fast decoding algorithms,
our main contribution here is to show that LSH good lattices do in fact exist.

From the perspective of the complexity of
ANN, LSH-good lattices (when given as advice to an ANN algorithm) provide a
slight improvement over~\cite{conf/focs/AndoniIndyk06} when using any of the
recent $2^{k+o(k)}$-time and $2^{k+o(k)}$-space algorithms for the closest vector
problem~\cite{conf/soda/BD15,conf/focs/ADS15} to implement the hash queries. In
particular, for $(c,r)$-ANN on an $n$ element database in $\R^d$, by choosing
the dimension of the lattice to be $k = \log^{2/3}(n)$, we get query time
$dn^{\rho}$ using $dn + n^{1+\rho}$ space where $\rho = 1/c^2 +
O(1/\log^{1/3}(n))$. These complexity results for ANN are however superseded by
the more recent approaches using \emph{data dependent}
LSH~\cite{conf/soda/AINR14,conf/stoc/AR15}, which achieve $\rho =
1/(2c^2-1)+o(1)$. While more sophisticated, these approaches still depend on
rather impractical and expensive random space partitions -- with query
complexity $2^{O(\sqrt{d})}$ instead of $2^d$ -- and hence there is still room
for progress.  

Given this, we view our contribution mainly as a conceptual one, namely that
\emph{structured space partitions} can be optimal. We hope that this provides
additional motivation for developing space partitions which admit fast query
algorithms, and in particular for finding novel classes of ``spherical'' lattices
(LSH-good or otherwise) admitting fast CVP solvers. We note that up to present,
the only known general classes of lattices for which CVP is solvable in
polynomial time are lattices of Voronoi's first kind
(VFK)~\cite{jour/siamdm/MGC14} and tensor products of two root
lattices~\cite{jour/decocrypt/DW17}, whose geometry is still rather restrictive
(see~\cite{thesis/Vallentin03} section 2.3 for an exposition of VFK lattices).\\

%\paragraph
%\noindent
\subsection{{\bf Techniques and High Level Proof Plan}}
The main techniques we use come from the theory of random lattices in the
geometry of numbers. While getting precise estimates on an LSH collision
probability for a generic high dimensional lattice seems very
difficult, it turns out to be much easier to estimate the average collision
probability for \emph{random lattices}. The distribution on lattices we use is
known as the Siegel measure on lattices, which is an invariant probability
measure on the space of lattices of determinant $1$ whose existence was
established by Siegel~\cite{jour/aom/Siegel45} (the invariance is with respect
to linear transformations of determinant $1$).

%Siegel used this measure to give
%  an alternate proof of the Minkowski-Hlawka theorem, which gives asymptotic lower
%  bounds (i.e.~as a function of dimension) on the density of lattice packings of
%  hyperspheres and more generally star-shaped domains.  \knote{Could remove the
%  last sentence without losing the reader.}

A powerful point of leverage when using random lattices drawn from the Siegel
distribution is that one can compute expected lattice point counts using
volumes. In particular, for any Borel set $S \subseteq \R^k$, we have the useful
identity $\E_L[|(L \cap S) \setminus \{0\}|] = {\rm vol}(S)$, i.e.~the expected
number of non-zero lattice points in $S$ is equal to its volume. We will need
more refined tools than this however, and in particular, we shall rely heavily
on powerful probabilistic estimates of Schmidt~\cite{jour/pams/Schmidt58} and
Rogers~\cite{jour/plms/Rogers58} developed for the Siegel measure. More
specifically, Schmidt~\cite{jour/pams/Schmidt58} provides extremely precise
estimates on the probability that a Borel set of small volume does not intersect a
random lattice, while Rogers~\cite{jour/plms/Rogers58} gives similarly precise
estimates for the relative fraction of cosets of a random lattice not
intersecting a Borel set. 

Using these estimates, we quickly derive clean and tight integral expressions
for the average collision probabilities. From then on, the strategy is simple if
rather tedious, namely, to get precise enough estimates for these integrals to
be able to show that the average ``near'' collision probability to the power
$c^2+o(1)$ is larger than average ``far'' collision probability. With this
inequality in hand, we immediately deduce the existence of an LSH-good lattice
from the probabilitic method. To prove that a random lattice is in fact LSH-good
with high probability (making our proof constructive) it would suffice to show
concentration for the relevant collision probabilities. While this seems very
plausible, we leave it for future work. \\

%\paragraph
\noindent{{\bf Estimating the Collision Probabilities and the LSH Constant.}}
We now give a more detailed geometric explanation of what the collision
probabilities represent, how the computations for lattices differ from those for
a random ball partition, and how the random lattice estimates mentioned above
come into play.

We recall the lattice LSH family 
going from $\R^d$ to $\R^k$ induced
by a lattice $L \subset \R^k$. We shall assume here that $L$ has determinant
$1$ and hence that the Voronoi cell $\mathcal{V}_L$ of $L$ has volume $1$ (any
region that tiles space with respect to $L$ has the same volume). A function from the
hash family 
$\mathcal{H}$ is generated as follows. First, pick a uniform random coset $t
\leftarrow \R^k / L$ and a matrix $M \in \R^{k \times d}$ with i.i.d. $N(0,1/k)$
entries (i.e.~Gaussian with mean $0$ and variance $1/k$). On query $q$, we
define the hash value as $CV_L(Mq+t)$, namely the closest vector in $L$ to
$Mq+t$. Note that $M$ is normalized here to approximately preserve distances,
since $\E[\|Mq\|^2] = \|q\|^2$. For $x,y \in \R^d$, $\|x-y\|_2 = \Delta$, we
wish to estimate the collision probability 
\begin{equation}
\label{eq:col-prob}
p_{\Delta} := \Pr_{h \leftarrow \mathcal{H}}[h(x) = h(y)] = \Pr_{M,
t}[CV_L(t+Mx)=CV_L(t+My)] \text{ ,}
\end{equation}
where $M,t$ are as above. We will show shortly that the right hand side indeed
only depends on $\Delta$. Using the above hash family, showing that $L$ achieves
the optimal LSH exponent for an approximation factor $c > 0$ corresponds to showing 
\begin{equation}
\label{eq:opt-lsh-const}
\min_{\Delta > 0} \ln(1/p_\Delta)/\ln(1/p_{c\Delta}) \leq 1/c^2+o(1) \text{ .}
\end{equation}
Note that for any desired distance threshold $r > 0$, we can always scale the
database so that the scaled distance threshold becomes the minimizer above.
Clearly, to be able to get a good upper bound on the LHS of \eqref{eq:opt-lsh-const}, we have to
be able to derive tight estimates for the collision probability curve
$p_\Delta$ over a reasonably large range. 

To understand $p_\Delta$, we now show that the collision probability can be
expressed as the probability that a uniformly sampled point in $\mathcal{V}_L$
stays inside $\mathcal{V}_L$ after a Gaussian perturbation of size $\Delta$. Let
$x,y,M,t$ be as in~\eqref{eq:col-prob}. A first easy observation is that
conditioned on any realization of $M(y-x)$, the distribution of $Mx+t$ is still
uniform over cosets of $\R^n / L$ since $t$ is uniform. Therefore,
\begin{align*}
\Pr_{M,t}[CV_L(t+Mx)=CV_L(t+My)] 
&= \Pr_{M,t}[CV_L(t)=CV_L(t+M(y-x))] \\
&= \Pr_{t,g \leftarrow N(0,I_k/k)}[CV_L(t)=CV_L(t+\Delta g)] \\ 
&\quad \left(\text{ since $M(y-x)$ has distribution $N(0,\Delta^2I_k/k)$ } \right) \\
&= \Pr_{v \leftarrow \mathcal{V}_L, g \leftarrow N(0,I_k/k)}[v+\Delta g \in
\mathcal{V}_L] \text{ .}
\end{align*}
For the last equality, note first that the Voronoi cell contains exactly one
element from every coset of $\R^k/L$ and hence a uniformly chosen point $v$ from
$\mathcal{V}_L$ is also uniform over cosets. Lastly, by construction $CV_L(v) =
0$ and hence $CV_L(v) = CV_L(v+\Delta g) \Leftrightarrow v+\Delta g \in
\mathcal{V}_L$.

At this point, without any extra information about $\mathcal{V}_L$, the task of
bounding the delicate function of collision probabilities seems daunting if
not intractable (note that generically $\mathcal{V}_L$ is a polytope with
$2(2^k-1)$ facets). To compare with the ball partitioning approach, it is not
hard to show that up to a factor $2$, the collision probabilities there are in
correspondance with the quantities 
\[
q_\Delta := \Pr_{u \leftarrow r_k B_2^k,g \leftarrow N(0,I_k/k)}[u+\Delta g \in r_k B_2^k], 
\]
where $r_k \approx \sqrt{k/(2\pi e)}$ is the radius of a ball of volume $1$ in
$\R^k$. We use the volume $1$ ball here to make the correspondance to
$\mathcal{V}_L$ which also has volume $1$. Thus, to match the collision
probabilities of the ball, which we know yield the right exponent, one would
like $\mathcal{V}_L$ to ``look like'' a ball. Unfortunately, even seemingly strong
notions of sphericality, such as assuming that $\mathcal{V}_L$ is within a factor
$2$ scaling of a ball (which random lattices in fact satisfy,
see~\cite{journ/tit/ELZ05} for an exposition), do not seem to suffice to
estimate these delicate collision probabilities at the right ranges. Note that
to make the effects of the inevitable estimation errors and dimensionality
effects small in the minimization of~\eqref{eq:opt-lsh-const}, we will want both
$p_\Delta$ and $p_{c\Delta}$ to be quite small when we estimate the ratio of
their logarithms.  For the ball, the function $q_\Delta$ has the form
$e^{-\alpha \Delta^2}$, where $\alpha := \alpha(\Delta)$ varies slowly within a
constant range for $\Delta = O(\sqrt{k})$. Note that if $\alpha$ were in fact
constant, then $\ln(1/q_\Delta)/\ln(1/q_{c\Delta})$ would equal $1/c^2$ for
every $\Delta$. The region where $\alpha$ is the most stable turns out to be
around $\Delta = k^{1/4}$, where $q_\Delta$ is quite small, i.e.~around $e^{-\Omega(\sqrt{k})}$.

Fortunately, while computing precise estimates for a fixed $L$ is hard,
computing the average collision probability over the Siegel measure on the space
of lattices of determinant $1$ is much easier. Note that the expected collision
probability curve $\E_L[p_\Delta]$, where $L$ is chosen from the Siegel measure,
corresponds exactly to the collision probability curve associated with a slight
modification of the LSH family examined above, namely, where instead of using a
fixed lattice, we simply sample a new lattice $L$ from the Siegel measure for
each hash function instantiation.  We now argue that to show existence of a good
LSH lattice one can simply replace the collision probability curve above
$p_\Delta$ by the expected collision probability curve $\E_L[p_\Delta]$.  To see
this, assume that~\eqref{eq:opt-lsh-const} holds for the expected curve. By
rearranging, this implies that that there exists $\Delta
> 0$ such that $\E_L[p_{\Delta}]^{c^2-o(1)} \geq \E_L[p_{c\Delta}]$. Since
$c^2-o(1) \geq 1$, by Jensen's inequality
\begin{equation}
\label{eq:lsh-ineq}
\E_L[p_{\Delta}^{c^2-o(1)}] \geq \E_L[p_{\Delta}]^{c^2-o(1)} \geq \E_L[p_{c\Delta}] \text{ .}
\end{equation}
Thus, by the probabilistic method, there must exist a lattice $L'$ such that
$p_{\Delta}^{c^2-o(1)} \geq p_{c\Delta}$ holds for $L'$, which shows that $L'$
achieves an LSH constant of $1/c^2+o(1)$, as needed.

We now explain how one can compute the expected collision probabilities using the
estimates of Schmidt and Rogers. For a fixed $\Delta$, a direct computation reveals
\begin{equation}
\label{eq:exp-col-prob}
\begin{split}
\E_L[p_\Delta] 
 &= \E_{L,u \leftarrow \mathcal{V}_L, g \leftarrow N(0,I_k/k)} [u+\Delta g \in
\mathcal{V}_L] \\
 &= \E_{L,g \leftarrow N(0,I_k/k)}\left[\int_{\R^n} I[u \in \mathcal{V}_L,
u+\Delta g \in \mathcal{V}_L] {\rm du}\right] \quad \left(\text{ since
$\mathcal{V}_L$ has
volume $1$ }\right) \\
 &= \int_{\R^n} \Pr_{L,g\leftarrow N(0,I_k/k)}[u \in \mathcal{V}_L,
u+\Delta g \in \mathcal{V}_L] {\rm du} \text{ .} 
\end{split}
\end{equation}
Define $B_x$ for $x \in \R^k$ to be the open ball around $x$ of radius
$\|x\|$. Note that for a fixed $g$ and $u$, the event that both $u$ and $\Delta
g+u$ are in $\mathcal{V}_L$, can be directly expressed as $(B_u \cup B_{\Delta g+u}) \cap L =
\emptyset$, i.e.~that there is no lattice point closer to $u$ and $\Delta g+u$ than $0$. 
Thus, one can express~\eqref{eq:exp-col-prob} as
\begin{equation}
\label{eq:col-prob-2}
\int_{\R^n} \Pr_{L,g\leftarrow N(0,I_k/k)}[(B_u \cup B_{\Delta g + u}) \cap L = \emptyset] {\rm du} \text{ .} 
\end{equation}
From here, for fixed $g$ and $u$, the inner expression is exactly the
probability that a random lattice $L$ doesn't intersect a Borel set and hence we
may apply Schmidt's estimates. Here Schmidt shows that as long as the $B_u \cup
B_{\Delta g + u}$ has volume less than $k-1$, then under a mild technical assumption, we
can estimate 
\[
\Pr_L[(B_u \cup B_{\Delta g + u}) \cap L = \emptyset] \approx e^{-V_{u,\Delta g}}
\]
where $V_{u,\Delta g}$ is the volume of $B_u \cup B_{\Delta g + u}$. This estimate is
only useful when $u$ has norm roughly $r_k$, since otherwise the volume of $B_u$
is too large to usefully apply Schmidt's estimate. However, one would expect
that for large $u$, the probability that $u$ is in the Voronoi cell is already
quite small. This is formalized by Roger's estimate, which gives that the fraction
of cosets of $L$ that are not covered by the ball of volume $k$ around the
origin (i.e.~again radius roughly $r_k$) is approximately $e^{-k}$. In
particular, this implies that at most an $e^{-k}$ expected fraction of the
Voronoi cell (since points in the Voronoi cell are in one to one correspondance
with cosets) lies outside a ball of radius $\approx r_k$, and hence we can
truncate the integral expression~\eqref{eq:col-prob-2} at roughly this radius
without losing much. 

After these reductions, we get that the collision probabilities can be tightly
approximated by the following explicit integral:
\begin{equation}
\label{eq:col-prob-3}
\int_{\R^n} \E_g[e^{-V_{u,\Delta g}}] {\rm du} \text{.} 
\end{equation}

The proof now continues with an unfortunately very long and tedious calculation,
which shows that the above estimate closely matches the corresponding collision
probability $q_\Delta$ for the ball, thus yielding the desired LSH constant.

%\paragraph{{\bf Related Work.}} 
\subsection{Related Work}
As mentioned earlier, the works~\cite{conf/soda/AINR14,conf/stoc/AR15}
show how to use a data dependent version of LSH to give an improved ANN
exponent of $1/(2c^2-1)$, which was shown to be optimal under an appropriate
formalization of data dependence in~\cite{arxiv/AR15}. These works reduce ANN in
$\ell_2$ to ANN on the sphere via a recursive clustering approach, where the
base case of the recursion roughly corresponds to the clustered vectors being
embedded as nearly orthogonal vectors on the sphere. A generic reduction from
$\ell_2$ ANN to spherical ANN (without the exact base case guarantee as above) was
also given in earlier work of Valiant~\cite{conf/focs/Valiant12}. We note that
the above clustering style reductions to the sphere remain relatively 
impractical, and thus there still seems to be room for more direct and practical
$\ell_2$ methods. For a different vein, the
works~\cite{conf/soda/BDGL16,conf/soda/Christiani17,conf/soda/ALRW17} studied
the achievable tradeoffs between query time and space usage, where the optimal
tradeoff for hashing based approaches was achieved in~\cite{conf/soda/ALRW17}.

With respect to structured and practical LSH hash functions,
\cite{conf/nips/AILRS15} computed the collision probabilities for cross-polytope LSH on
the sphere (first introduced by~\cite{jour/ads/TT07,conf/sigkdd/ER08}), which
corresponds to a Voronoi partition on the sphere induced by a vertices of a
randomly rotated cross-polytope. As their main result, they show that when near vs
far corresponds to $\ell_2$ distance $\sqrt{2}/c$ vs $\sqrt{2}$ (the latter case
correspondings to orthogonal vectors), cross polytope LSH achieves the optimal
limiting exponent of $1/(2c^2-1)$, corresponding to the base case of the
recursive clustering approaches above. Furthermore, they show a fine grained
lower bound on the LSH exponent (when the far case again corresponds to
orthogonal vectors) of any hash function which partitions the sphere into at
most $T$ parts\footnote{Under the mild technical assumption that each piece
covers at most $1/2$ the sphere.}, which allows them to conclude that any
spherical LSH function that substantially improves upon cross polytope LSH needs
to have query time \emph{sublinear} in the number of parts. It is tempting here
to seek an analogy with lattice based LSH, in that the complexity of CVP
computations on a $d$-dimensional lattice $L$, after appropriate preprocessing,
can be bounded by $\tilde{O}(d^{O(1)} |\mathcal{V}_L|)$~\cite{conf/soda/BD15}
where $|\mathcal{V}_L|$ denotes the number of facets of the Voronoi cell of $L$.
Thus, one may wonder if $|\mathcal{V}_L|$ can be associated with the number of
``parts'' in an analogous manner.  For a generic $d$-dimensional lattice, we
note that $|\mathcal{V}_L| = 2(2^d-1)$, and thus the corresponding question would
be to find an LSH-good lattice for which CVP takes $\tilde{O}(2^{(1-\eps)d})$
for some positive $\eps > 0$. As another interesting comparison, the
$d$-dimensional cross polytope induces a partion with $2d$ parts whose gap to
optimality (in terms of the spherical LSH exponent) is $O(\log\log d/\log d)$,
whereas a random $d$-dimensional lattice has a Voronoi cell is $2(2^d-1)$ facets
with a gap to optimality (for $\ell_2$ LSH) of $O(1/\sqrt{d})$.  \\

%\paragraph
\subsection{{\bf Conclusions and Open Problems}}
To summarize, for a fixed approximation factor $c > 1$, we show that random
space partitions induced by \emph{shifts of a single lattice} can achieve the
optimal \emph{data oblivious} LSH exponent for the $\ell_2$ metric. While this
shows that we can hope for ``well-structured'' space partitions for $\ell_2$,
the lattices we use to show existence are \emph{random}, and are in many ways
devoid of easy to exploit structure (at least algorithmically). Thus, a natural
open question is whether one can find a more structured family of lattices
achieving the same limiting LSH exponent for which CVP queries can be executed
faster. In terms of improving the present result, another natural question would
be to make our proof constructive and to show that for a fixed dimension $k$,
there exists a single $k$-dimensional lattice which achieves the optimal LSH
exponent for every $c \geq 1$. \\

%\paragraph
\noindent{{\bf Organization.}} 
In Section~\ref{sec:prelims}, we setup notations and define formally the notion of lattices and approximate nearest neighbor search problem. 
We describe our lattice based hash function family in Section~\ref{sec:hash_function} and analyze its performance. The helper theorems needed to show the main result are proved in subsequent sections.

\section{Preliminaries}\label{sec:prelims}
\newcommand{\cV}{{\cal V}}

We denote the set $\{1, 2, \ldots, n\}$ by $[n]$. We work over the Euclidean space. For $x\in \R^d$, let  $||x||=\sqrt{\sum_i x_i^2}$ denote the $\ell_2$ norm of $x$.
Let  $V_B$ denote the volume of a $k$-dimensional unit-radius ball.  Let $\tau=\sqrt{k}\cdot  V_B^{\frac 1k}$. By standard geometry facts, $\tau = \sqrt{2 \pi e} \left( 1 + O\left( \frac 1k \right) \right)$. For $x\in \R^k$,  let $B_x$ denote the open ball centered at $x$ of radius $\|x\|$ and let $V_x$ denote its volume. Note that $V_x=V_B\|x\|^k$.

%\paragraph
\noindent\textbf{Lattices.} A lattice $L\subset \R^d$ is the set of all linear combinations with integer coefficients of a set of linearly independent vectors $\{b_1, b_2, \ldots, b_r\}$, i.e., $L=\{\sum_i \alpha_i b_i\mid \alpha_i\in\Z ~\forall i\in [r]\}$. The lattice may be represented by the $d\times r$ basis matrix $B$, whose columns are the vectors $b_i$. 
 If the {\em rank} $r$ is exactly equal to $d$, then the lattice is said to have {\em full rank}. It is common to assume that the lattice has full rank, and we do so in what follows, since otherwise one may just work over the real span of $B$. 
 
 The {\em quotient group} $\R^d/L$ of $L$ is the set of cosets $c+L=\{c+v\mid v\in L\}$, where $c\in \R^d$, with the group operation $(c_1+L)+(c_2+L)=(c_1+c_2)+L$. 
The {\em determinant} of $L$, denoted $\det(L)$, is defined as $\det(L)=\sqrt{B^T B}$.  
A lattice has multiple bases: if $B$ is a basis then $B U$ is also a basis, for any unimodular matrix $U$ (i.e., a matrix $U$ with integer entries with $\det(U)=1$.)
 %The {\em fundamental parallelepiped} with respect to $B$ is the set ${\cal P}(B):=\{Bx:x\in [-1/2,1/2)^d\}$.
The {\em Voronoi cell} of a lattice is the set of all points closer to the origin than to any other lattice point. Formally, 
${\cV}_L:=\{x \in \R^d\mid ||x||\leq||x-v||, \forall v\in L\}.$ 
Define the {\em shifted} Voronoi cell centered at $v$, denoted $\cV_{L}(v)$, to be the set of points $v+{\cV}_L=\{v+u\mid u\in {\cV}_L \}$. It is a standard fact that the set of cells $\{v+{{\cV}_L}\}_{v\in L}$ cover the entire space $\R^d$. Moreover, for every $x\in \R^d$, there exists a $v\in L$ such that $x-v\in {\cV}_L$. 
In fact, %both the (half-open) fundamental parallelepiped and 
the (half-open) Voronoi cell %are {\em fundamental domains} of $L$: sets that 
contains exactly one representative from each coset $c+L$, for $c\in \R^d$.
One of the fundamental computational problem on lattices is the {\em Closest Vector Problem} (CVP) defined as follows: given a target vector $t\in \R^d$, find a closest vector from the lattice $L$ to $t$. We will denote a solution to CVP with input $t$ by $CV_L(t)$.
  We will use recent algorithms running in time $O(2^d)$ as a blackbox \cite{conf/focs/ADS15}. 
We will need the following property of the Voronoi cell.

\begin{fact}\label{fact:v-t}
$v \in CV_L(t)$ if and only if $t-v\in {\cV}_L$.
\end{fact}

%\subsection{Approximate Near Neighbor and LSH}

%\paragraph
\noindent\textbf{Approximate Near Neighbor and LSH.} In the $c$-approximate near neighbor($c$-ANN) problem, given a collection ${\cal P}$ of $n$ points in $\R^d$, and  parameters $r, \delta>0$, the goal is to construct a data structure with the following property: on input a query point $q\in \R^d$, with probability $1-\delta$, if there exists $p\in {\cal P}$ with $||q-p||\leq r$, it outputs some point $p'\in {\cal P}$, with $||q-p'||\leq c\cdot r$. By a simple scaling of the coordinates, one may assume that $r=1$. Also, $\delta$ is assumed to be a constant, and the success probability can be amplified by building several instances of the data structure.

A family ${\cal H}$ is a {\em locality-sensitive hashing} scheme  with parameters $(1, c, p_1, p_2)$ if it satisfies the following properties: for any $p, q\in \R^d$
\begin{itemize}
\item if $||p-q||\leq 1$ then $\Pr_{\cal H}[h(q)=h(p)]\geq p_1$,
\item if $||p-q||\geq c$ then $\Pr_{\cal H}[h(q)=h(p)]\leq p_2$.
\end{itemize}

The initial work of \cite{conf/stoc/IM98} shows that an LSH scheme implies a data structure for $c$-ANN.

\begin{theorem}\cite{conf/stoc/IM98}
Given a LSH family ${\cal H}$ with parameters $(1, c, p_1, p_2)$, where each function in ${\cal H}$ can be evaluated in time $\tau$, let $\rho=\frac{\log(1/p_1)}{\log(1/p_2)}$. Then there exists a data structure for $c$-ANN with $O((d+\tau)n^\rho \log_{1/p_2} n)$ query time, using $O(dn+n^{1+\rho})$ amount of space.
\end{theorem}

%\paragraph
\noindent \textbf{Multidimensional Gaussian.}
A $d$-dimensional Gaussian distribution with mean $0$ and covariance matrix $\sigma^2 I_d \in \R^{d\times d}$ 
has density function $$p(x)=\frac{1}{(2\pi)^{d/2} \sigma^d} \exp(-\frac{||x||^2}{2\sigma^2}),$$ and is denoted by $N(0,\sigma^2 I_d)$.

\section{Our Lattice-based Hash Family and Proof Strategy}\label{sec:hash_function}
%\paragraph
\noindent\textbf{LSH family for lattice $L$ with $\det(L)=1$.} A hash function $h=h_{M, t}$  indexed by a projection matrix $M \in \R^{k \times d}$  from $\R^d$ to $\R^k$, and a vector $t\in \R^k$ is constructed as follows:
\begin{enumerate}
\item  pick the entries $M_{i,j}$  according to a Gaussian distribution with mean $0$ and variance $1/k$. %$k^{-\frac 12}$
\item  pick $t$ uniformly from the Voronoi cell ${\cal{V}}_L$ of $L$ (centered at $0$). Sampling $t$ can be achieved by sampling from $\R^k/L$, namely by sampling from the fundamental parallelepiped with respect to any basis. 
\end{enumerate}

Given a point $a\in \R^d$, we define  $h(a)$ to be a closest vector in $L$ to its projection $ Ma$ translated by $t$. Formally, 
\[ h(a) = CV_L(Ma+t).\]

We first show that for $a, b\in \R^d$ the quantity  $\Pr_{M, t} [h(a) = h(b)]$ only depends on the distance $||a-b||$, and not on the points $a, b$ themselves.

\begin{prop}\label{prop:p_delta}
 Let $a, b \in \R^d$ be arbitrary and let $\Delta=||a-b||$. Then 
 \[ \Pr_{M, t }[h(a) = h(b) ]  = \Pr_{x\leftarrow \mathcal{V}_L, y\leftarrow N(0,\Delta^2I_k/k)}[ x+y \in {\cal V}_L]. \] 
\end{prop}
Let $p_{\Delta}$ denote the probability of collision of two inputs which are exactly distance $\Delta$ apart. i.e., $p_{\Delta} := \Pr_{M, t }[h(a) = h(b) ]$, where $||a-b||=\Delta$. An easy argument shows that $p_{\Delta}$ is non-increasing 
as a function of $\Delta$.

\begin{corollary}\label{cor:pdelta}
$p_{\Delta}$ is non-increasing as a function of $\Delta$. 
\end{corollary} 

The performance of our LSH family is measured by the LSH constant defined by 
\[\rho_L:= \min_{\Delta > 0} \frac{\ln 1/p_{\Delta}}{\ln 1/p_{c\Delta}}. \]

Our result shows the existence of a lattice $L$ with optimal performance guarantee.
\begin{restatable}{theorem}{thmmainlsh}
\label{thm:main_lsh}
 For every $k$ large enough and $c > 1$, there exists a $k$-dimensional lattice $L$ with $\det(L)=1$  achieving 
$$\rho_L \leq \frac{1}{c^2} + O\left(\frac{1}{\sqrt{k}}\right). $$
\end{restatable}

Theorem \ref{thm:main_lsh} follows from our main technical result, which bounds the expected collision probabilities $p_{\Delta}$ and $p_{c\Delta}$ for $\Delta = k^{1/4}$.

\begin{restatable}{theorem}{thmmaincollisionpr}
\label{thm:main_collision_pr}
%Let $K$ be a large enough constant. %For every $ 1< c \leq k^{\frac 14} / K$, 
For every $k$ large enough and $c > 1$, 
there exist absolute constants $K_1, K_2, K_3$ such that for $\Delta = k^{1/4}$,  
\begin{align*}
&\E_L\left[p_{\Delta} \right] \geq K_1~e^{- \frac{\tau^2}{8} \sqrt{k}}
\qquad  \mbox{ and, } \\
&\E_L\left[ p_{c\Delta} \right] \leq K_2 ~e^{- \frac{\tau^2}{8} c^2 \sqrt{k} \left(1 - \frac{K_3 c^2}{\sqrt{k}} \right)},  
\end{align*}  where the expectation is over $k$-dimensional lattices $L$ with $\det(L)=1$.
\end{restatable}

We can now prove Theorem~\ref{thm:main_lsh} using Corollary~\ref{cor:pdelta} and Theorem~\ref{thm:main_collision_pr}. 
%%%%%%%%%%%%%%%%%%%%%%%%
%Proof of thm:main_lsh
%\thmmainlsh*
\begin{proof}[Proof of Theorem~\ref{thm:main_lsh}]
For any $\Delta > 1$, define $\tilde{\rho} := \frac{\ln 1/\E_L[p_{\Delta}]}{\ln 1/\E_L[p_{c\Delta}]}$. From Corollary~\ref{cor:pdelta}, we know that $p_{\Delta}$ is non-increasing. Hence, $\tilde{\rho} \leq 1$ for any $c > 1$.  So, we can use Jensen's inequality to get that
\begin{align*}
\E_L\left[ p_{\Delta}^{1/\tilde{\rho}} \right]
&\geq \E_L\left[ p_{\Delta} \right]^{1/\tilde{\rho}}  \qquad \mbox{ (Jensen's inequality) }\\
&=  \E_L\left[ p_{c\Delta} \right]  \qquad \mbox{ (by the definition of
$\tilde{\rho}$) }.
\end{align*}

By the probabilistic method, it then follows that there exists a $k$-dimensional lattice $L$ with $\det(L)=1,$ such that the collision probabilities satisfiy $\frac{\ln 1/p_{\Delta}}{\ln 1/p_{c\Delta}} = \tilde{\rho}$ and hence, 
%Therefore, there exists a $k$-dimensional lattice $L$ with $\det(L)=1,$ such that $
$\rho_L \leq \tilde{\rho}$. 

We now show that $\tilde{\rho} \leq \frac{1}{c^2} + O\left( \frac{1}{\sqrt{k}} \right)$.  
From Theorem~\ref{thm:main_collision_pr} we know that for any $c > 1$, 
%$ 1< c \leq k^{\frac 14}/K$ 
and $\Delta = k^{\frac 14}$, there exist constants $K_1, K_2, K_3$ such that 
\begin{align*}
&\E_L\left[p_{\Delta} \right] \geq K_1~e^{- \frac{\tau^2}{8} \sqrt{k}}
\qquad  \mbox{ and, } \\
&\E_L\left[ p_{c\Delta} \right] \leq K_2 ~e^{- \frac{\tau^2}{8} c^2 \sqrt{k} \left(1 - \frac{K_3 c^2}{\sqrt{k}} \right)}
\end{align*}

Note that for $c > \frac{k^{\frac 14}}{\sqrt{K_3}}$, the upper bound on $\E_L\left[ p_{c\Delta} \right]$ from Theorem~\ref{thm:main_collision_pr} becomes trivial. 
First, we consider the case when $c  \leq  \frac{k^{\frac 14}}{2 \sqrt{K_3}}$. %where $K' = \max \{2 \sqrt{K_3}, K \}$. 
For this value of $c$, we can use bounds obtained in Theorem~\ref{thm:main_collision_pr} to show that $\tilde{\rho} \leq \frac{1}{c^2} + O\left(\frac{1}{\sqrt{k}}\right)$ as follows:
\begin{align*}
\frac{\ln 1/\E_L(p_{\Delta})}{\ln 1/\E_L(p_{c\Delta})} 
&\leq  \frac{\frac{\tau^2}{8} \sqrt{k} - \ln K_1 }{\frac{\tau^2}{8} c^2 \sqrt{k} \left(1 - \frac{K_3 c^2}{\sqrt{k}} \right)  - \ln K_2} \\
&\leq \frac{1}{c^2}\left( 1 + K_4 \frac{c^2}{\sqrt{k}}\right) \qquad \mbox{ for some constant $K_4$ }.
\end{align*}

Now, for $c > \frac{k^{\frac 14}}{2 \sqrt{K_3}}$, we need to show that there exists a $k$-dimensional lattice of determinant 1, such that $\rho_L \leq \frac{1}{c^2} + O\left( \frac{1}{\sqrt{k}} \right) $. From the monotonicity of $p_{\Delta}$, we know that for any $c' < c$, 
$p_{c \Delta} \leq p_{c' \Delta}$. Therefore, consider $c' = k^{\frac 14} /2 \sqrt{K_3}  < c$. From Theorem~\ref{thm:main_collision_pr}, and the analysis above, we know that there exists a lattice of determinant $1$ such that 
\begin{align*}
\rho_L &\leq \frac{1}{c'^2}\left( 1 + K_4 \frac{c'^2}{\sqrt{k}}\right) \qquad \mbox{ for some constant $K_4$ }\\
&= \frac{2K_3}{\sqrt{k}} + \frac{K_4}{\sqrt{k}} \\
&= \frac{1}{c^2} + O\left( \frac{1}{\sqrt{k}} \right).
\end{align*}
\end{proof}
%%%%%%%%%%%%%%%%%%%%%%%%
%%End proof of thm:main_lsh

Proving  Theorem~\ref{thm:main_collision_pr} poses substantial technical hurdles. We will break the proof into smaller components, which we describe after introducing some helpful notation.

%%%%%% Notations %%%%%%%
%Recall that $V_B$ denotes the volume of a $k$-dimensional unit-radius ball.  $\tau=\sqrt{k}\cdot  V_B^{\frac 1k}$.  
%For $x\in \R^k$,  $B_x$ denotes the open ball centered at $x$ of radius $\|x\|$ and $V_x$ denotes its volume. We recall that $V_x=V_B\|x\|^k$.

For any $\Delta \geq 1 $, define 
\[ I(\Delta^2) := \int_{x \in \R^k: V_x \leq \frac k8}  \E_{y\leftarrow N(0, \Delta^2I_k/k)} \left[e^{ - V_x -V_{x+y}}\right] ~dx. \] 
In the next lemma, we show tight bounds on $\E_L[ p_{\Delta} ]$ in terms of 
$I(\Delta^2)$.
%%%%%%%%%%%%%%%%%%%%%%%%%%%%%%%%%%%%%%
% 2 main theorem statements - together prove thm:main_collision_pr

\begin{restatable}{lemma}{thmboundexpcoll}
\label{lem:bound_exp_coll} 
For every $k$ large enough and any $\Delta \geq 1$, 
\begin{align*}
I(\Delta^2)  - e^{-k/8} 
 ~\leq~ \E_L[ p_{\Delta} ]  
 ~\leq~ 4 I(4^{-\frac 2k} \Delta^2) + 3e^{-k/8}.
\end{align*}
where the expectation is over $k$-dimensional lattices $L$ with $\det(L)=1$.
\end{restatable}

We now show tight bounds for $I(\Delta^2)$ for $\Delta^2 = \beta \sqrt{k}$ ,where $1 \leq \beta \leq O(\sqrt{k})$  in Lemma~\ref{lem:main}, which is the most technically delicate part of the analysis, as it involves precise balancing of parameters and taking care of minutious details.
\begin{lemma}\label{lem:main}
%Let $K$ be a constant. For any $ 1 \leq \beta \leq K \sqrt{k}$,
There exist absolute constants $K \in [0, 1], K_1, K_2, \bar{K}_1, \bar{K}_2$ such that for any $1 \leq \beta \leq K \sqrt{k}$, 
\[
\bar{K}_1 ~ e^{- \frac{\tau^2}{8} \beta \sqrt{k} \left(1 + \frac{\bar{K}_2 \beta}{\sqrt{k}} \right) }
~\leq~ I\left(\beta \sqrt{k}\right)
~\leq~ K_1~e^{- \frac{\tau^2}{8} \beta \sqrt{k} \left( 1 -  \frac{K_2 \beta}{\sqrt{k}} \right) }
\]
\end{lemma}
%%%%%%%%%%%%%%%%%%%%%%%%%%%%%%%%%%%%%%

We now show how  Lemmas~\ref{lem:bound_exp_coll} and Lemma~\ref{lem:main} imply Theorem \ref{thm:main_collision_pr}.

%%%%%%%%%%%%%%%%%%%%%%%%%%%%%%%%%%%%%%
% Proof of thm:main_collision_pr
\begin{proof}[Proof of Theorem~\ref{thm:main_collision_pr}]
First we prove the lower bound on $\E_L[p_{\Delta}]$ for $\Delta = k^{\frac14}$. 
From Lemma~\ref{lem:bound_exp_coll} and Lemma~\ref{lem:main}, we have %that for some constant $\bar{K}$, 
\begin{align*}
\E_L[p_{\Delta}] &\geq I(\Delta^2) - e^{-k/8} \qquad \mbox{ (from Lemma~\ref{lem:bound_exp_coll}) }\\
&\geq \bar{K}_1 ~ e^{- \frac{\tau^2}{8} \sqrt{k} \left(1 + \frac{\bar{K}_2}{\sqrt{k}}  \right) } - e^{-k/8} \qquad \mbox{ (from Lemma~\ref{lem:main} with $\beta = 1$) } \\
&\geq \bar{K}_3 ~ e^{- \frac{\tau^2}{8} \sqrt{k}}.
\end{align*}

Similarly, for the upper bound on $\E_L[p_{c\Delta}]$ for $\Delta = k^{\frac14}$, we get 
\begin{align*}
\E_L[p_{c \Delta}]  & \leq 4 ~I(4^{-\frac 2k} c^2 \Delta^2) + 3e^{-k/8}  \qquad \mbox{ (from Lemma~\ref{lem:bound_exp_coll}) }\\
& \leq K_1~e^{-  4^{-\frac 2k} c^2 \frac{\tau^2}{8} \sqrt{k} \left( 1 -  \frac{K_2 c^2}{\sqrt{k}}  \right) } + 3e^{-k/8}  \qquad \mbox{ (from Lemma~\ref{lem:main} with $\beta = 4^{-\frac 2k} c^2 $) } \\
& \leq K_3 ~e^{-c^2 \frac{\tau^2}{8} \sqrt{k} \left(1 - \frac{K_2 c^2}{\sqrt{k}} \right)} \qquad \mbox{ (since $4^{-\frac 2k} \geq 1 - O(1/k)$) } .
\end{align*}
Note that since Lemma~\ref{lem:main} holds for $\beta < O(\sqrt{k})$, the upper bound on $\E_L[p_{c \Delta}]$ holds for $c^2 \leq K \sqrt{k}$ for some constant $K$.
\end{proof}
% END Proof of thm:main_collision_pr
%%%%%%%%%%%%%%%%%%%%%%%%%%%%%%%%%

We conclude this section with the proof of Proposition~\ref{prop:p_delta} and of Corollary~\ref{cor:pdelta},  while devoting the rest of the paper for the proof of Lemma~\ref{lem:bound_exp_coll}. 
Due to space constraints, the proof of Lemma \ref{lem:main} will appear in the full version of the paper. 

%%%%%%%%%%%%%%%%%%%%%%%%%%%%%%%%%%%%%%
% Proof of prop:p_delta
\begin{proof}[Proof of Proposition~\ref{prop:p_delta}]
Let $M$ and $t$ be as defined above. From the definition of the hash function, $h(a)=h(b)$  if $Ma + t$ and $Mb + t$ land in the same Voronoi cell of $L$ about some lattice point. Let $||a-b||=\Delta$. We have 
\begin{eqnarray}\label{eq:p_delta}
\begin{split}
p_{\Delta} &= \Pr_{M, t} [h(a) = h(b)]  \\
&= \Pr_{M, t}\left[CV_L(Ma+t) = CV_L(Mb +t) \right] \\
&= \Pr_{M, t}[ Ma+t, Ma+M(b-a)+t  \text{ lie in the same Voronoi cell }].
\end{split}
\end{eqnarray}
Let $Ma + t \in {\cal V}_{L}(\ell)$ for some $\ell \in L$. Define $x := Ma + t - \ell \in {\cal V}_L$. 
Note that because of the random shift $t$, $x$ is a uniform random point in the Voronoi cell of $L$ about $0$. 

Let $y: = M(b-a) \in \R^k$. 
Since each entry $M_{ij}$ of $M$ is a Gaussian random variable with 0 mean and variance $1/k$, therefore, the $i^{th}$ entry of $y$, given as $y_i = \sum_{j=1}^k M_{ij}(b_j-a_j)$ has mean $0$ and variance $\frac 1k \sum_j (b_j - a_j)^2 = \frac{\Delta^2}{k}$.

Plugging these observations in Equation~\ref{eq:p_delta}, we get
\begin{align*}
p_{\Delta} &= \Pr_{M, t}[Ma+t -\ell, Ma+M(b-a)+t - \ell \in {\cal V}_L] \\
&= \Pr_{x \leftarrow{\cal V}_L , y \leftarrow N(0, \Delta^2 I_k / k)}[ x, x+y \in {\cal V}_L] \\
&= \Pr_{x\leftarrow{\cal V}_L, y \leftarrow N(0, \Delta^2 I_k / k) }[x+y \in  {\cal V}_L].
\end{align*}

\end{proof}
% END Proof of prop:p_delta
%%%%%%%%%%%%%%%%%%%%%%%%%%%%%%%%%%%%%%

%%%%%%%%%%%%%%%%%%%%%%%%%%%%%%%%%%%%%%\
% Proof of cor:pdelta
\begin{proof}[Proof of Corollary~\ref{cor:pdelta}]
By Proposition~\ref{prop:p_delta}, it suffices to show that the function 
\[
f(s) = \Pr_{x\leftarrow {\cal V}_L,y\leftarrow N(0,I_k/k)}[x+ s y \in {\cal V}_L] \text{,}
\]
where $x$ is uniform in ${\cal V}_L$ and $y$ is standard Gaussian, is a 
non-increasing 
function of $s$ on $\R_+$. Since ${\cal V}_L$ has volume $1$ and $x+sy \in {\cal
V}_L \Leftrightarrow x \in {\cal V}_L - sy$, we have that
\[
\Pr_{x,y}[x+ sy \in {\cal V}_L] = \Pr_y[{\rm vol}({\cal V}_L \cap ({\cal V}_L - sy)] \text{ .}
\]
Define $g_y(s) := {\rm vol}({\cal V}_L \cap ({\cal V}_L - sy))$. We claim that
$g_y(s)$ is non-decreasing on $(-\infty,0]$ and non-increasing on $[0,\infty)$. To
see this, note that by symmetry of ${\cal V}$, $g_y$ is symmetric, i.e. $g_y(s) =
g_y(-s)$. Furthermore, for $\lambda \in [0,1]$, $s_1,s_2 \in \R$,  
\begin{align*}
g_y(\lambda s_1 + (1-\lambda) s_2)^{1/n} 
&= {\rm vol}({\cal V}_L \cap ({\cal V}_L - \lambda (s_1 + (1-\lambda) s_2)y))^{1/n} \\
&\geq {\rm vol}(\lambda ({\cal V}_L \cap ({\cal V}_L - s_1 y)) +
(1-\lambda) ({\cal V}_L \cap ({\cal V}_L - s_2 y)))^{1/n} \\
& \qquad \qquad \qquad \left(\text{ by containment } \right) \\
&\geq \lambda {\rm vol}({\cal V}_L \cap ({\cal V}_L - s_1 y))^{1/n}
 + (1-\lambda){\rm vol}({\cal V}_L \cap ({\cal V}_L - s_2 y))^{1/n} \\
& \qquad \qquad \qquad \left(\text{ by Brunn-Minkowski } \right) \\
&= \lambda g_y(s_1)^{1/n} + (1-\lambda) g_y(s_2)^{1/n} \text{ .}
\end{align*}
Therefore, $g_y(s)^{1/n}$ is a symmetric, non-negative and concave function of
$s$. Any symmetric concave function on $\R$ must attain its maximum value at
$0$, and hence must be non-increasing away from $0$. 

Now consider $0 \leq s_1 \leq s_2$. Since $g_y$ is non-increasing on
$\R_+$, we get that
\[
f(s_1) = \E_y[g_y(s_1)] \geq \E_y[g_y(s_2)] = f(s_2)
\]
as needed.

\end{proof}
% END Proof of cor:pdelta
%%%%%%%%%%%%%%%%%%%%%%%%%%%%%%%%%%%%%%

%PART-1: Rogers / Schmidt
\section{Proof of Lemma~\ref{lem:bound_exp_coll} }\label{sec:proof_part_1}
In the previous section, we had seen that the expected collision probability between points which are $\Delta$ apart is defined as 
\begin{align*}
\E_L[ p_{\Delta} ] &=  \E_L\left[ \Pr_{x\leftarrow \mathcal{V}_L, y\leftarrow N(0,\Delta^2I_k/k)}[ x+y \in {\cal V}_L] \right]\\
&= \int_{x \in \R^k} \int_{y \in \R^k} \Pr_L(x, x+y \in {\cal V}_L) \cdot \frac{ e^{-\frac{\|y\|^2}{2 \sigma^2}} }{ \left( 2 \pi \sigma^2 \right)^{\frac{k}{2}}} ~dy~ dx
 \qquad \mbox{ for }  \sigma^2 = \Delta^2 / k. 
 \end{align*}

The goal of this section is to derive tight bounds for this expression through the proof of Lemma~\ref{lem:bound_exp_coll}. 

Recall that $B_x$ denotes the open $k$-dimensional ball centered at $x \in \R^k$ of radius $\|x\|$ and 
$B_{x+y}$ denotes the open $k$-dimensional ball centered at $x+y \in \R^k$ of radius $\|x+y\|$. Also, $V_x$ and $V_{x+y}$ denotes their volumes.  Consider $B_{x,y} = B_x \cup B_{x+y}$, the union of $B_x$ and $B_{x+y}$ and let $V_{x,y}$ denote its volume.
We will need the following theorem for the proof of Lemma~\ref{lem:bound_exp_coll}.

\begin{lemma}\label{lem:schmidt_ap}
\[
e^{-V_{x,y}} - e^{- k/4} ~\leq~ \Pr_L(x, x+y \in {\cal V}_L) ~\leq~  e^{-\frac12 V_{x,y}} + e^{- k/4}.
\]
\end{lemma}

%%%%%%%%%%%%%%%%%%%%%%%%%%%%%%%%%%%%%%%
In order to prove Lemmas~\ref{lem:bound_exp_coll} and \ref{lem:schmidt_ap}, we invoke the following results of Rogers \cite{jour/plms/Rogers58} and Schmidt \cite{jour/pams/Schmidt58}. 

\begin{theorem}[Corollary of \cite{jour/plms/Rogers58}, Theorem 1]\label{thm:rogers}
Let $B$ be the $k$-dimensional ball of volume $V$ centered at the origin.
If $V \leq \frac k8$,  then there exists a constant $k_0$ such that for $k > k_0$, 
\[
\left|~\int_{x \in \R^k} \Pr_{L} \left[ x \in {\cal V}_L \setminus B \right]~dx - e^{-V}~\right| \leq  c_1k^3 \left( \frac{16}{27}\right)^{\frac k4}
\]
where, the probability is taken over the space of all lattices of determinant $1$.   
\end{theorem}

\begin{theorem}[\cite{jour/pams/Schmidt58}, Theorem 4]\label{thm:schmidt}
Let $S$ be a Borel set of measure $V$ such that $0 \notin S$ and for all $x \in S$, $-x \notin S$. If $V \leq k-1$,  then for $k \geq 13$, 
 
\[
\Pr_L\left[ L \cap S = \emptyset \right] = e^{-V}\left(1 - R\right).
\]
where, the probability is taken over the space of all lattices of determinant $1$ and 
$\lvert R  \rvert < 6 \left( \frac 34 \right)^{\frac k2} e^{4V} + V^{k-1}k^{-k + 1} e^{V + k}$. 
\end{theorem}

%%%%%%%%%%%%%%%%%%%%%%%%%%%%%%%%%%%%%%%
\begin{fact}
\[ \frac12 \left( V_x + V_{x+y} \right) \leq V_{x,y} \leq V_x + V_{x+y}.\]
\end{fact}
\begin{proof}
Let WLOG, $V_x \leq V_{x+y}$. Also, we know that $V_{x,y} = V_x + V_{x+y} - V(B_x \cap B_{x+y})$. We now show that $V(B_x \cap B_{x+y}) \leq \frac12 \left( V_x + V_{x+y} \right) $. This fact follows easily from the observation that  the intersection volume is at most the volume of the smaller ball. Therefore, 
\begin{align*}
V(B_x \cap B_{x+y}) &\leq V_x = \frac12 V_x + \frac12 V_x \leq \frac12 \left( V_x + V_{x+y} \right).
\end{align*}
\end{proof}
%%%%%%%%%%%%%%%%%%%%%%%%%%%%%%%%%%%%%%%
%Proof of thm:bound_exp_coll
We now prove Lemma~\ref{lem:bound_exp_coll} using Lemma~\ref{lem:schmidt_ap}.
\begin{proof}[Proof of Lemma~\ref{lem:bound_exp_coll} ]
For notational convenience, we will use $\sigma^2$ to denote $\Delta^2 / k$. From the definition of $p_{\Delta}$ and Proposition~\ref{prop:p_delta}, we have 
\begin{eqnarray}\label{eq:exp_coll_pr}
\begin{split}
\E_L[ p_{\Delta} ] &= \E_L\left[ \Pr_{x\leftarrow \mathcal{V}_L, y\leftarrow N(0,\sigma^2 I_k)}[ x+y \in {\cal V}_L] \right]\\
&= \int_{x \in \R^k} \int_{y \in \R^k} \Pr_L(x, x+y \in {\cal V}_L) \cdot \frac{ e^{-\frac{\|y\|^2}{2 \sigma^2}} }{ \left( 2 \pi \sigma^2 \right)^{\frac{k}{2}}} ~dy~ dx\\
&=\int_{x \in \R^k : V_x \leq \frac k8} \int_{y \in \R^k } \Pr_L(x, x+y \in {\cal V}_L) \cdot \frac{ e^{-\frac{\|y\|^2}{2 \sigma^2}} }{ \left( 2 \pi \sigma^2 \right)^{\frac{k}{2}}} ~dy~ dx \\
&\qquad \qquad +\int_{x \in \R^k: V_x > \frac k8 } \int_{y \in \R^k} \Pr_L(x, x+y \in {\cal V}_L) \cdot \frac{ e^{-\frac{\|y\|^2}{2 \sigma^2}} }{ \left( 2 \pi \sigma^2 \right)^{\frac{k}{2}}} ~dy~ dx.
\end{split}
\end{eqnarray}
We first note that if $V_x \geq \frac k8$, then the probability that $x\in {\cal V}_L$ is itself very small. This fact gives us tight bounds on $\E_L[ p_{\Delta} ] $ up to additive $e^{-\Omega(k)}$ term. We use  
Theorem~\ref{thm:rogers} to formalize this statement. Let $B_{0}$ be  the $0$ centered ball of volume $\frac k8$. We have,
\begin{align*}
\int_{x \in \R^k : V_x \geq \frac k8} & \int_{y \in \R^k} \Pr_L(x, x+y \in {\cal V}_L) \cdot \frac{ e^{-\frac{\|y\|^2}{2 \sigma^2}} }{ \left( 2 \pi \sigma^2 \right)^{\frac{k}{2}}} ~dy~ dx \\
&\leq \int_{x \in \R^k : V_x \geq \frac k8} \int_{y \in \R^k} \Pr_L(x \in {\cal V}_L) \cdot \frac{ e^{-\frac{\|y\|^2}{2 \sigma^2}} }{ \left( 2 \pi \sigma^2 \right)^{\frac{k}{2}}} ~dy~ dx\\
&= \int_{x \in \R^k : V_x \geq \frac k8} \Pr_L(x \in {\cal V}_L) dx\\
&= \int_{x \in \R^k} \Pr_L(x \in {\cal V}_L \setminus B_{0}) ~dx \\
&= e^{-\frac k8} + e^{-\frac k8}. \qquad (\mbox{ from Theorem~\ref{thm:rogers}  })
\end{align*}

Plugging this observation into the expression for $\E_L[p_{\Delta}]$ in Equation~\ref{eq:exp_coll_pr}, we get that
\begin{align*}
\int_{x \in \R^k : V_x \leq \frac k8} & \int_{y \in \R^k } \Pr_L(x, x+y \in {\cal V}_L) \cdot \frac{ e^{-\frac{\|y\|^2}{2 \sigma^2}} }{ \left( 2 \pi \sigma^2 \right)^{\frac{k}{2}}} ~dy~ dx \\
& \qquad \qquad \leq \E_L[ p_{\Delta} ] \\
& \leq \int_{x \in \R^k : V_x \leq \frac k8} \int_{y \in \R^k } \Pr_L(x, x+y \in {\cal V}_L) \cdot \frac{ e^{-\frac{\|y\|^2}{2 \sigma^2}} }{ \left( 2 \pi \sigma^2 \right)^{\frac{k}{2}}} ~dy~ dx  + 2e^{- k/8}.
\end{align*}

Further, using the bounds on $\Pr_L(x, x+y \in {\cal V}_L)$ from Lemma~\ref{lem:schmidt_ap}, we get 
\begin{align*}
& \int_{x \in \R^k: V_x \leq \frac k8} \int_{y \in \R^k} \left( e^{-V_{x,y}}  - e^{- k/4} \right)\cdot  \frac{ e^{-\frac{\|y\|^2}{2 \sigma^2}} }{ \left( 2 \pi \sigma^2 \right)^{\frac{k}{2}}} ~dy~ dx \\
&\qquad \qquad \leq \E_L[ p_{\Delta} ] \\
&\leq  \int_{x \in \R^k: V_x \leq \frac k8} \int_{y \in \R^k} \left( e^{-\frac12 V_{x,y}} + e^{- k/4}\right)\cdot \frac{ e^{-\frac{\|y\|^2}{2 \sigma^2}} }{ \left( 2 \pi \sigma^2 \right)^{\frac{k}{2}}} ~dy~ dx   + 2e^{- k/8}. 
\end{align*}

Since $V_{x,y} \leq V_x + V_{x+y}$, the lower bound in the theorem statement then follows trivially. 
\begin{align*}
\E_L[ p_{\Delta} ]  &\geq \int_{x \in \R^k: V_x \leq \frac k8} \int_{y \in \R^k} \left( e^{-V_{x,y}}  - e^{- k/4} \right)\cdot  \frac{ e^{-\frac{\|y\|^2}{2 \sigma^2}} }{ \left( 2 \pi \sigma^2 \right)^{\frac{k}{2}}}~dy~ dx \\
&= \int_{ \substack{x \in \R^k\\ V_x \leq \frac k8} } \int_{y \in \R^k} e^{-V_{x,y}} \frac{ e^{-\frac{\|y\|^2}{2 \sigma^2}} }{ \left( 2 \pi \sigma^2 \right)^{\frac{k}{2}}}~dy~ dx -  \int_{ \substack{x \in \R^k\\V_x \leq \frac k8} } \int_{y \in \R^k}  e^{- k/4} \frac{ e^{-\frac{\|y\|^2}{2 \sigma^2}} }{ \left( 2 \pi \sigma^2 \right)^{\frac{k}{2}}}~dy~ dx \\
&\ge \int_{x \in \R^k: V_x \leq \frac k8} \E_{y\sim N(0, \sigma^2~I_k)} \left[ e^{-V_x - V_{x+y}}\right]~dx - \frac k8 e^{- k/4} \\
&\ge \int_{x \in \R^k: V_x \leq \frac k8} \E_{y\sim N(0, \sigma^2~I_k)} \left[ e^{-V_x - V_{x+y}}\right]~dx - e^{- k/8} 
\end{align*}

For the upper bound, set $u = 4^{-\frac 1k} x$, and $v = 4^{-\frac 1k} y$. Since $\frac12 V_{x,y} \geq \frac{V_x +V_{x+y}}{4} = V_u + V_{u+v}$, we have  
\begin{align*}
\E_L[ p_{\Delta} ] 
&\leq  \int_{x \in \R^k: V_x \leq \frac k8} \int_{y \in \R^k} \left( e^{-\frac12 V_{x,y}} + e^{- k/4} \right) \cdot \frac{ e^{-\frac{\|y\|^2}{2 \sigma^2}} }{ \left( 2 \pi \sigma^2 \right)^{\frac{k}{2}}}~dy~ dx   + 2e^{- k/8} \\
&\leq \int_{x \in \R^k} \int_{y \in \R^k} e^{-\frac{V_x +V_{x+y}}{4}} \frac{e^{-\frac{\|y\|^2}{2\sigma^2}}}{\left( 2\pi \sigma^2\right)^{\frac{k}{2}}} ~dy~dx  + 3e^{- k/8} \\
&= \int_{u \in \R^k} \int_{v \in \R^k} e^{-V_u - V_{u+v}} \frac{e^{-\frac{\|v\|^2~4^{\frac{2}{k}}}{2\sigma^2}}}{\left( 2\pi \sigma^2\right)^{\frac{k}{2}}} ~4 dv~ 4 du + 3e^{- k/8} \\
&=  4 \int_{u \in \R^k} \int_{v \in \R^k} e^{-V_u - V_{u+v}} \frac{e^{-\frac{\|v\|^2}{2 (4^{-\frac{1}{k}}\sigma)^2}}}{\left( 2 \pi (4^{-\frac{1}{k}}\sigma)^2\right)^{\frac{k}{2}}} ~dv~du  + 3e^{- k/8}\\
&=  4 \int_{u \in \R^k}  \E_v \left[ e^{-V_{u}-V_{u+v}} \right]~du + 3e^{- k/8} ~~\text{ where, } v \sim N(0, 4^{-\frac 2k} \sigma^2~I_k).
\end{align*}
\end{proof}
%END Proof of thm:bound_exp_coll
%%%%%%%%%%%%%%%%%%%%%%%%%%%%%%%%%%%%%%%

Now it remains to prove Lemma~\ref{lem:schmidt_ap}.

\begin{proof}[Proof of Lemma~\ref{lem:schmidt_ap}]
Recall that $B_{x,y} = B_x \cup B_{x+y}$, the union of $B_x$ and $B_{x+y}$ and $V_{x,y}$ denotes its volume. 
We note that $x$ and $x+y$ are in the voronoi cell of a lattice $L$ if and only if $B_{x,y}$ does not contain any lattice points. Therefore, 
\begin{align*}
\Pr_L\left[ x, x+y \in {\cal V}_L \right] &= \Pr_L\left[ B_{x,y} \cap L = \emptyset \right] 
\end{align*}

As a first case, suppose $V_{x,y} < \frac{k}{32}$.  Now consider the following partition of $B_{x,y}$. Let $S$ be the set of points $a \in B_{x,y}$ such that $-a \in B_{x,y}$. 
\[ S = \{ a \in B_{x,y} \mid -a \in B_{x,y} \}. \]
Partition $S$ with respect to an arbitrary hyperplane as follows: Define $S_1 = \{ a \in S \mid a^t y < 0 \}$ and $S_2 = S \setminus S_1$ for an arbitrarily chosen $y\in \R^k$. Note that for every $a \in S_1$, $-a \in S_2$. Define $A = (B_{x,y} \setminus S) \cup S_1$. Note that $\{ A, S_2 \}$ is a partition of $B_{x,y}$, i.e., $B_{x,y} = A \cup S_2$, and $A \cap S_2 = \emptyset$. 

Without loss of generality, assume that $A$ is the larger partition of $B_{x,y}$, i.e $V_A \geq \frac12 V_{x,y}$.  
Also from the definition of $A$ and $S_2$, we have that if $A \cap L = \emptyset$, then $S_2 \cap L = \emptyset$. We can now apply Theorem~\ref{thm:schmidt} for both $A$ and $S_2$. 
\begin{align*}
\Pr_L\left[ B_{x,y} \cap L = \emptyset \right] 
& =  \Pr_L\left[  (A \cap L = \emptyset  ), (S_2 \cap L  = \emptyset) \right] \\
&= \Pr_L\left[ A \cap L = \emptyset  \right] ~\Pr_L\left[  (S_2 \cap L = \emptyset  ) \mid  (A \cap L  = \emptyset) \right]\\
&= \Pr_L\left[ A \cap L = \emptyset  \right]\\
&=e^{-V_{A}}\left(1 - R_{A} \right) \qquad \mbox{ where, } \lvert R_{A} \rvert = 6 \left( \frac 34 \right)^{\frac k2} e^{4V_A} + V_A^{k-1}k^{-k + 1} e^{V_A + k}.
\end{align*}
Since $\frac12 V_{x,y} \leq V_{A} \leq V_{x,y} < \frac{k}{32}$, we have 
$\lvert R_{A} \rvert < e^{- k/4}$.
Therefore, 
\[
e^{-V_{x,y}}\left( 1- e^{- k/4} \right) \leq \Pr_L\left[ B_{x,y} \cap L = \emptyset \right]  \leq e^{-\frac 12V_{x,y}} \left( 1 + e^{- k/4} \right).
\]

Next, suppose $V_{x,y} > \frac{k}{32}$. Then consider a body $B_{x,y}'$ contained in $B_{x,y}$ of volume $\frac{k}{32}$. Using a similar argument as above with $B_{x,y}$ replaced with $B_{x, y}'$, we conclude that
\[
\Pr_L\left[ B_{x,y} \cap L = \emptyset \right] \leq \Pr_L\left[ B_{x,y}' \cap L = \emptyset \right]  \leq e^{- k/4}.
\]

\end{proof}

\bibliographystyle{plainurl}
\bibliography{refLSH}

\end{document}